\documentclass[12pt]{article}
\usepackage[intlimits]{amsmath}
\usepackage{amssymb,amsthm,slashed,mathabx}

\usepackage[T1]{fontenc}
\usepackage[cp1250]{inputenc}
\usepackage{cite}

\usepackage[showonlyrefs]{mathtools}
\mathtoolsset{showonlyrefs=true}
\usepackage{fixltx2e}
\MakeRobust{\eqref}

\usepackage{bbm}

\newcommand{\Hc}{\mathcal{H}}
\newcommand{\mR}{\mathbb{R}}
\newcommand{\mC}{\mathbb{C}}
\newcommand{\lc}{C_+}
\newcommand{\lct}{C_+^t}
\newcommand{\as}{\mathrm{as}}
\newcommand{\inc}{\mathrm{in}}
\newcommand{\out}{\mathrm{out}}
\newcommand{\ret}{\mathrm{ret}}
\newcommand{\adv}{\mathrm{adv}}
\newcommand{\rad}{\mathrm{rad}}
\newcommand{\tr}{\mathrm{tr}}
\newcommand{\wt}{\widetilde}
\newcommand{\dl}{d^2l}
\newcommand{\al}{\alpha}
\newcommand{\ga}{\gamma}
\newcommand{\ep}{\epsilon}
\newcommand{\la}{\lambda}
\newcommand{\w}{\omega}
\newcommand{\W}{\Omega}
\newcommand{\V}{\dot{V}}
\newcommand{\vep}{\varepsilon}
\newcommand{\dsp}{\displaystyle}
\newcommand{\txt}{\textstyle}
\newcommand{\ov}{\overline}
\newcommand{\p}{\partial}
\newcommand\df{:=}
\newcommand{\cc}{\mathrm{compl.\,conj.}}

\DeclareMathOperator{\Rp}{Re}
\DeclareMathOperator{\Ip}{Im}
\DeclareMathOperator{\id}{\mathbbm{1}}

\newtheorem{pr}{Proposition}

\title{Asymptotic structure of electrodynamics revisited}
\author{Andrzej Herdegen\thanks{e-mail: herdegen@th.if.uj.edu.pl}\\
{\it Institute of Physics, Jagiellonian University,}\\
{\it ul.\,S.\,{\L}ojasiewicza 11, 30-348  Krak\'{o}w, Poland}}\date{}

\begin{document}

\maketitle

\begin{abstract}
We point out that recently published analyses of null and timelike infinity and long-range structures in electrodynamics to large extent rediscover results present in literature. At the same time, some of the conclusions these recent works put forward may prove controversial. In view of these facts we find it desirable to revisit the analysis taken up more than two decades ago, starting from earlier works on null infinity by other authors.

\vspace{1ex}
\noindent
keywords: asymptotic electrodynamics, memory effects, large gauge transformation

\vspace{1ex}

\end{abstract}

\section{Introduction}\label{int}

In a series of recent articles -- see \cite{kps15} and \cite{cl15} and works cited there -- a~group of authors report on new conservation laws and symmetries both in classical and quantum electrodynamics. These structures are supposed to be encoded in the asymptotic properties of electrodynamics. Later, an analog of this structure in gravitation theory was used to analyze anew the black hole information paradox \cite{hps16}. Also, in the context set by these investigations there is a noticeable interest in literature in so-called memory effects, both in gravitation, as well as electrodynamics (see \cite{kps15}, \cite{pa15} and works cited there).

In this article we want to comment on the electromagnetic part of this programme in view of results known from literature. We shall point out that the structure of asymptotic electrodynamics, along lines similar to those proposed by the authors mentioned above, has received more extensive attention in the past. On the other hand, we shall also point out that some claims in their programme may prove controversial. We use this opportunity to summarize the programme taken up by the present author more than two decades ago.

We sketch our main points and the plan of the article.

(i) On the classical level a more extensive analysis of the asymptotic structure of electrodynamics is proposed in the article \cite{he95}, which extends many ideas present earlier in literature, in particular the null infinity analysis by Bramson \cite{br77} and Ashtekar and Streubel \cite{ast81}. The null infinity ``matching conditions'' (so called by Strominger et al. \cite{kps15}) were known before, and in \cite{he95} they are derived in full generality. After recalling an important tool in the field of the wave equation in Section \ref{measure}, we discuss the classical electrodynamics, and the asymptotics of fields in Sections \ref{homog} -- \ref{timelike}. Invariant structures in the asymptotic regions are introduced in Section \ref{invstruct}.

(ii) Our discussion consequently uses Lorenz gauge potentials in electrodynamics. These potentials have well-defined dynamics and null asymptotes appropriate for the asymptotic structure. In Section \ref{lgtrans}, still on the classical level, we ask what other gauges are admitted by this structure. We point out that although the gauge transformation called by the authors of \cite{kps15} and \cite{cl15} (and others) `large gauge transformation' -- LGT in the following -- produces potentials with asymptotes admitted by the structure, it is not a symmetry of this structure. Therefore, the changes induced by LGT \emph{may be} incorporated within this structure, but with interpretation wholly different from a~gauge transformation: they modify the \emph{electromagnetic fields} of the state.

(iii) Commenting in Section \ref{memory} on the relations between long-range structure and so-called memory effects, we recall that the earliest effect of this type  in electrodynamics was discovered long ago by Staruszkiewicz in 1981 \cite{st81}. What concerns us here, Staruszkiewicz's effect is a clear argument in favor of the observability of the long-range degrees of freedom in electrodynamics.

(iv) The classical asymptotic structure was quantized and analyzed extensively in a series of papers \cite{he98, he05, he08, hr11}. `Asymptotic quantization' by Ash\-te\-kar \cite{as86} may  be viewed as the reference starting point of this analysis, but the discussion goes much further towards asymptotic algebraic formulation of matter-radiation system. The authors mentioned at the beginning, as it seems, may have motivations similar to my own: they treat variables at spatial infinity as genuine quantum observables. We make general remarks on specific infrared problems of quantum electrodynamics in Section \ref{quantum}, and then summarize the construction of our asymptotic quantum theory in Sections \ref{constalg} -- \ref{repr}.

(v) We construct a unitary operator built from the long-range quantum variables at spacelike infinity in Section \ref{splike} . We indicate that the construction of such an operator demands an explicit use of a concrete representation of the theory and, within that representation, a subtle limiting. We show that this unitary operator induces a~nontrivial transformation of the basic quantum variables of the asymptotic algebra. This transformation contains elements corresponding to the classical effect induced by LGT. As in the classical case, the changes modify the physics of the state, and their interpretation within the limits of the asymptotic theory is far from a gauge symmetry. We show in Appendix \ref{identity} that our rigorously constructed operator corresponds, in loose terms, to the `LGT generator' of \cite{kps15} and \cite{cl15} (which is not rigorously defined).

(vi) However, if a scattering matrix may be defined, then it is true that the operators constructed by smearing of long-range variables should be conserved. This is a quantum version of a `continuous conservation' law for long-range tails in classical electrodynamics. We mention earlier statements on this problem in Section \ref{scatt}.

(vii) Appendix contains some more technical material needed in the main text.

The article is based mainly on the articles \cite{he95, he98, he05, he08, hr11}. We give a specific reference only in some special cases.

\section{Invariant measure on the set of lightlike directions}\label{measure}

The first fact worth wider knowledge is the existence of a Lorentz-invariant measure on the set of light-like directions. The existence of this measure is a~classical result \cite{ggv66, zw76, st81, pr84}, but -- strangely enough -- ignored by most physicists despite its crucial role in the field of the wave equation and massless fields.\pagebreak

We consider the flat spacetime with a fixed origin, thus described by the Min\-kow\-ski vector space $M$ with the signature $(+,-,-,-)$. We shall write $l$ for any future-pointing lightlike (nonzero) vector and denote
\begin{equation}
 \lc=\{\,l\mid l\cdot l=0,\ l^0>0\,\}\,.
\end{equation}
We shall write $t$ for any future-pointing timelike, unit vector, and denote
\begin{equation}
 \lct=\{\,l\in\lc\mid t\cdot l=1\,\}
\end{equation}
(a unit sphere in the hyperplane $t\cdot x=1$). It is well-known that the measure $d\mu_0(l)=d^3l/l^0$ on the lightcone is Lorentz-invariant. However, equally important, but less known, is the following construction. Let $f(l)$ be a measurable function on $\lc$, homogeneous of degree $-2$: for each $\ga>0$ there is $f(\ga l)=\ga^{-2}f(l)$. Then the integral defined by
\begin{equation}\label{d2l}
 \int f(l)\,d^2l=\int_{\lct} f(l)\,d\W_t(l)\,,
\end{equation}
where $d\W_t(l)$ is the angle measure on the unit sphere, does not depend on the choice of the vector $t$ (Lorentz invariance\footnote{In fact, the measure is not only Lorentz, but generally conformal invariant, see \cite{ggv66, pr84} and Appendix \ref{intnull}.}). We give a simple proof of this fundamental fact in Appendix \ref{intnull}.  It follows that if we denote ($a$, $b$, etc.\ are spacetime indices\footnote{Usually, we treat spacetime indices in the spirit of `abstract index notation' of Penrose \cite{pr84}, but if concrete coordinates are needed (as, e.g., for asymptotic expansion), then always a Minkowski basis coordinates are meant.})
\begin{equation}
 L_{ab}=l_a\frac{\p}{\p l^b}-l_b\frac{\p}{\p l^a}
\end{equation}
-- the generators of the Lorentz transformations, intrinsic differential operators on the lightcone, then
\begin{equation}\label{Ld2l}
 \int L_{ab}f(l)\,d^2l=0\,.
\end{equation}

\section{Homogeneous Maxwell equations}\label{homog}

Let $A(x)$ be a Lorenz-gauge potential (we often suppress spacetime indices) of a free electromagnetic field $F_{ab}=\p_aA_b-\p_bA_a$, thus satisfying the homogeneous wave equation. The Fourier representation
\begin{equation}\label{free}
 A(x)=\frac{1}{\pi}\int e^{-ix\cdot k}a(k)\vep(k^0)\delta(k^2)d^4k\,,
\end{equation}
where $\delta$ is the Dirac measure, $\vep(\w)$ is the sign of $\w$, and where the relations
\begin{equation}
 \ov{a(k)}=-a(-k)\,,\quad k\cdot a(k)=0\,,\quad a(k)\rightarrow a(k)+k\beta(k)
\end{equation}
reflect reality of $A$, the Lorenz condition and a possible gauge transformation respectively, defines a very wide class (also possibly distributional) of such fields, and one needs a physically motivated selection criterion. An obvious restriction is to consider fields with finite energy-momentum, which is given by
\begin{equation}
 P^a=-\int \ov{a(k)}\cdot a(k)\, k^a\, d\mu_0(k)\,.
\end{equation}
This is still a very large class and its often considered subclass -- having mathematical advantages -- consists of the so called regular wave packets -- solutions with initial data (on a Cauchy hyperplane) having smooth Fourier transforms with compact support not containing the zero vector. However, this class does not include fields of the Coulomb-like decay at spatial infinity -- such fields are produced in scattering of charged particles. We describe further selection in three steps.

\subsection{Step 1: Coulomb-like spacial decay. Conservation of spacelike tails}\label{step1}

 The selection criterion which admits fields produced in scattering events, but at the same time leaves more infrared-singular solutions outside, is this: for $y^2<0$ the potential $A$ has a well-defined scaling limit
\begin{equation}\label{asx}
 A^\as(y)\df\lim_{\la\to\infty}\la A(\la y)\,,\qquad y^2<0\,.
 \end{equation}
In terms of the Fourier transform this is equivalent to the existence of the limit
\begin{equation}\label{asp}
  a^\as(k)=\lim_{\mu\searrow0}\mu\, a(\mu k)\,;
\end{equation}
note that $a^\as(l)$ is homogeneous of degree $-1$. \emph{This type of singularity of} $a(k)$ \emph{is consistent with the finite energy demand},\footnote{We emphasize this well-known fact, as the authors of \cite{kps15} seem not to appreciate it. We shall return to this point later on.}  and we further consider fields satisfying both restrictions. For fields satisfying our selection criterion not only the scaling limit centered at zero is well defined, but also for each $x\in M$ and $y\in M$, $y^2<0$, there is
\begin{multline}\label{aspot}
 \lim_{\la\to\infty}\la A(x+\la y)=\frac{-i}{2\pi}\int\frac{a^\as(l)}{y\cdot l-i0}\,\dl +\cc \\
 = \frac{1}{\pi}\int\frac{\Ip a^\as(l)}{y\cdot l}\,\dl+\int\Rp a^\as(l)\,\delta(y\cdot l)\,\dl\,.
\end{multline}
This may be viewed, if one wishes, as a `continuous conservation law': the form of the limit does not depend on the choice of the reference point $x$ in Minkowski space. We note that:
\begin{quote}
\emph{This `conservation law' has nothing to do with any gauge symmetry -- the field $A$ may be of any algebraic type -- and is only a property of the wave equation together with the selection criterion~\eqref{asp}}.
\end{quote}
Note also that the two parts on the rhs have definite parities with respect to the reflection $y\to-y$: the first part is odd, while the second is even.\footnote{We mention as an aside, that there exists an interesting quantum theory by Staruszkiewicz \cite{st89}, based on a slight extension of the structure appearing in \eqref{aspot}, which aims at geometrization of the elementary charge. See also other later works of this author, as well as \cite{he05}.}  Now, all standard free fields taking part in scattering processes in electrodynamics are of the second type: the spacelike tails of their Lorentz potentials are even (and the tails of the fields themselves are correspondingly odd). For example, the potential of the radiation field (i.e. $A^\ret-A^\adv$) of a point particle with charge $q$ scattered instantaneously in $x=0$ is given by
\begin{equation}\label{qrad}
 A^\mathrm{rad}(x)=q\theta(-x^2)\bigg(\frac{v}{\sqrt{(v\cdot x)^2-x^2}}
 -\frac{u}{\sqrt{(u\cdot x)^2-x^2}}\bigg)\,,
\end{equation}
where $v$ and $u$ are the initial and the final velocity respectively (the singularity at the cone $x^2=0$ is an artefact due to the instantaneous nature of the scattering event). The same is true for radiation fields produced by sources due to massive fields (as in the Maxwell-Dirac system). This absence of an odd part may be used as a further restricting criterion, but before doing this we want to give another form to the above formulae.

Let $l\in C_+$ and denote
\begin{equation}\label{Va}
 \V(s,l)=-\int\limits_{-\infty}^{+\infty} \w a(\w l) e^{-i\w s}\,d\w\,,
\end{equation}
where overdot denotes differentiation with respect to the real variable $s$, and the function $V$ will be obtained by integration. It is easy to see that $\V$ is a~real function homogeneous of degree $-2$, and -- in the case of a vector function representing a Lorenz potential -- transversal:
\begin{equation}
 \V(\mu s,\mu l) = \mu^{-2} \V(s, l)\quad  (\mu>0)\,,\qquad l\cdot \V(s,l)=0\,.
\end{equation}
Using this function in the Fourier representation of $A(x)$ one obtains another useful integral representation of the solution of the wave equation
\begin{equation}\label{freedl}
 A(x)=-\frac{1}{2\pi}\int\V(x\cdot l,l)\,\dl\,,
\end{equation}
whose advantage is that the integration goes over a compact set. If $A$ is a~Lorenz potential of an electromagnetic field, then the field is expressed by
\begin{equation}\label{freeFdl}
 F_{ab}(x)=-\frac{1}{2\pi}\int \big[l_a\ddot{V}_b(x\cdot l,l)-l_b\ddot{V}_a(x\cdot l,l)\big]\,\dl\,,
\end{equation}
and the gauge transformation takes now the form
\begin{equation}\label{gaugetr}
 \V(s,l)\rightarrow \V(s,l)+l\dot{\al}(s,l)\,.
\end{equation}
Now, our selection criterion -- the existence of $a^\as(k)$, Eq.\,\eqref{asp} --  implies that $\w\Rp a(\w l)$ is continuous in $\w$ at zero, while $\w\Ip a(\w l)$ has a jump at zero of the magnitude $2\Ip a^\as(l)$. Correspondingly, the behavior of $\V(s,l)$ for large $|s|$ is governed by \footnote{More precisely, the estimate of the rest in the form $O(|s|^{-1-\ep})$ is a slightly stronger assumption. If $\Ip a^\as(l)=0$, what we are going to assume in the next section, and $\ep\in(0,1)$, then this estimate implies that $\w a(\w l)$ is $\ep$-H\"older continuous in $\w$.}
\begin{equation}\label{estim}
 \V(s,l)=-\frac{2}{s}\Ip a^\as(l)+O(|s|^{-1-\ep})\quad
 \text{for}\quad |s|\to\infty\,.
\end{equation}

\subsection{Step 2. Null asymptotes}\label{step2}

The null asymptotic behavior of the field, which we now want to investigate, depends crucially on this estimate. If $\Ip a^\as(l)\neq0$, then $RA(x\pm Rl)$ is of order $\log R$ for large $R$. Moreover, one can show that also angular momentum radiated over finite time into any solid angle is not well-defined for \mbox{$R\to\infty$}. This is another reason to apply the more restrictive selection criterion mentioned before. Thus from now on we assume that
\begin{equation}\label{ima}
 \Ip a^\as(l)=0\,.
 \end{equation}
In that case the estimate \eqref{estim} implies that the function $\V(s,l)$ is absolutely integrable over $s\in\mR$ and we fix $V(s,l)$ by demanding that it vanishes for $s\to\infty$:
\begin{equation}\label{V+}
 V(s,l)=-\int_s^{+\infty}\V(s',l)\,ds'\quad\iff\quad V(+\infty,l)=0\,;
\end{equation}
then for $s\to-\infty$ it has a~well defined limit obtained from the inversion of formula \eqref{Va}:
\begin{equation}
 V(-\infty,l)=2\pi\lim_{\w\to0}\w a(\w l)\,.
\end{equation}
We also define
\begin{equation}\label{Vprim}
 V'(s,l)=-V(s,l)+V(-\infty,l)\,,
 \end{equation}
so that
\begin{equation}\label{V'-}
 V'(-\infty,l)=0\,,
\end{equation}
and then
\begin{equation}\label{freematch}
 V'(+\infty,l)=V(-\infty,l)\,.
\end{equation}
With these definitions the null asymptotic behavior of the potentials is easily expressed \cite{he95}:
\begin{equation}\label{freenullA}
 \lim_{R\to\infty}RA_b(x+Rl)=V_b(x\cdot l,l)\,,\qquad \lim_{R\to\infty}RA_b(x-Rl)=V'_b(x\cdot l,l)\,,
\end{equation}
and the null asymptotes of the electromagnetic fields are then:
\begin{gather}\label{freenullF}
 \lim_{R\to\infty}RF_{ab}(x+Rl)=l_a\V_b(x\cdot l,l)-l_b\V_a(x\cdot l,l)\,,\\
 \lim_{R\to\infty}RF_{ab}(x-Rl)=l_a\V'_b(x\cdot l,l)-l_b\V'_a(x\cdot l,l)\,.
\end{gather}
It is evident that the relations between $V$ and $V'$ are gauge-independent, and quantities $l\wedge V(s,l)$ and $l\wedge V'(s,l)$ are gauge invariant.

The real parameter $s$ appearing in our formulae has a twofold interpretation. On the one hand, if in the above asymptotic formulas one sets $x=st$ ($t$ as defined in Section \ref{measure}) and scales $l$'s to $t\cdot l=1$, then $s$ is the retarded or advanced time, according to the case. On the other hand, if one compactifies the spacetime \emph{a la} Penrose, then $s$ is an affine parameter along generators of future or past null infinity. In particular, $V'(+\infty,l)$ and $V(-\infty,l)$ are the values of the respective functions on the future (past) edge of the past (future) null infinity respectively, and their equality \eqref{freematch} is a free field version of a~general `matching property' to be discussed below. Fields with $V(-\infty,l)=0$ are called \emph{infrared-regular}, and for them the spacelike tail vanishes, as seen from the formula \eqref{aspot}. Fields with $V(-\infty,l)\neq0$ are called \emph{infrared-singular}.

At this point we want to stress once more a crucial fact concerning the definition of $V(s,l)$. As indicated above, this function is obtained by integrating $\V(s,l)$ over $s$, so it may appear one has a freedom of transformation $V(s,l)\rightarrow V(s,l)+V_0(l)$, with $V_0$ constant in~$s$. However, $A(x)$ is uniquely determined by $\V(s,l)$ alone according to the formula \eqref{freedl}, so its null asymptotes (past and future) are also unique. It is only with the specification \eqref{V+} (and \eqref{Vprim} for the past null asymptote, with \eqref{V'-} as a consequence) that the asymptotic formulae \eqref{freenullA} are true.

We exploit this knowledge to characterize more fully possible functions of gauge transformations. Suppose that two Lorenz potentials differ by a gauge transformation, $A^{(2)}=A^{(1)}+\p\Lambda$. Then $V^{(2)}(s,l)=V^{(1)}(s,l)+l\alpha(s,l)$, with $V^{(2)}(+\infty,l)=V^{(1)}(+\infty,l)=0$, thus also $\alpha(+\infty,l)=0$. According to formula~\eqref{freedl} we have $\p\Lambda(x)=-\frac{1}{2\pi}\int l\dot{\alpha}(x\cdot l,l)\,d^2l$, which implies
\begin{equation}
 \Lambda(x)=-\frac{1}{2\pi}\int \alpha(x\cdot l',l')\,d^2l'+\vep_+\,,
\end{equation}
where $\vep_+$ is a constant number. We set here $x=st\pm Rl$ and take the limit $R\to\infty$. As $l\cdot l'>0$ almost everywhere, we find
\begin{equation}\label{Lambdanull}
\lim_{R\to\infty}\Lambda(st\pm Rl)=\vep_\pm\,,\qquad \vep_-=\vep_+-\frac{1}{2\pi}\int\alpha(-\infty,l')\,d^2l'\,.
\end{equation}
We summarize:
\begin{quote}
\emph{Among Lorenz gauges there can be no gauge transformation $\p\Lambda$ whose asymptote $l\alpha(s,l)$ remains constant along null infinity generators (i.e., is independent of $s$). For any admissible gauge function $\Lambda$ the null-asymptotic limits \eqref{Lambdanull} are two, in general different, constant numbers}.\footnote{Note that although the field $\Lambda(x)$ is smooth, in general the past and future null asymptotic limits $\vep_\pm$ do not `match'.}
\end{quote}

To close this section we would like to illustrate the null asymptotic behavior, and comment on the possibility of asymptotic expansions of the form\footnote{An equivalent form is assumed in \cite{kps15}, formulas (2.8); the difference in expansion of $A_z$ results from the use of coordinate basis in which $z$-vector is proportional to $R$.}  $A(st+Rl)=\sum_{n=1}^\infty R^{-n}A^{(n)}(s,l)$, in the simplest possible case of a~scalar field~$A$, spherically symmetric in a chosen frame. Let $A$ be given by \eqref{freedl} with $\V(s,l)=(t\cdot l)^{-2}\dot{W}((t\cdot l)^{-1}s)$, where $t$ is a~fixed time axis vector and $\dot{W}(\tau)$ is any regular function decaying at least like $|\tau|^{-1-\ep}$ for $|\tau|\to\infty$. This field satisfies all our selection criteria, and the explicit integration in formula \eqref{freedl} gives
\begin{equation}
 A(x)=\frac{W\big(t\cdot x-\sqrt{(t\cdot x)^2-x^2}\big)-W\big(t\cdot x+\sqrt{(t\cdot x)^2-x^2}\big)}{\sqrt{(t\cdot x)^2-x^2}}\,,
\end{equation}
where  $W(\tau)=-\int_\tau^\infty\dot{W}(\tau')d\tau'$. One can now immediately explicitly confirm all our statements on the null asymptotes, Eqs \eqref{Vprim}--\eqref{freenullA} and subsequent comments. Let us choose two examples:
\begin{equation}
 W_1(\tau)=\frac{\alpha}{(\tau^2+\tau_0^2)^{\ep/2}}\,,\quad W_2(\tau)=\int\limits_\tau^{+\infty}\frac{\beta\,d\tau'}{({\tau'}^2+\tau_0^2)^{(1+\ep)/2}}\,,
\end{equation}
where $\alpha$, $\beta$ and $\tau_0$ are constants. With these choices both fields $A_1$ and $A_2$ fall into the class specified in this section, the first one is infrared regular, while the second is infrared singular. However, for $\ep<1$ the asymptotic expansion of these fields does not exist in the form of the series in $1/R$ beyond the leading order. In general, this problem is even more involved, but for the electromagnetic fields due to the potentials defined in the present section the two leading orders do exist.

\subsection{Step 3. Fields of electric type}\label{step3}

There is one final step needed to complete our selection criterion. The fact that there are no magnetic monopoles and free fields are thus produced by scattering electric charges is encoded in the property of the spatial tail of the field being of electrical type. This is equivalent, as it turns out, to the condition
\begin{equation}
 L_{[ab}V_{c]}(-\infty,l)=L_{[ab}V'_{c]}(+\infty,l)=0\,.
\end{equation}
Smooth functions, homogeneous of degree $-1$, transversal: $l\cdot V(-\infty,l)=0$, and satisfying the above condition form a linear space, whose properties are described in Appendix \ref{unchv}. It follows that
\begin{equation}\label{VFi}
 l_aV_b(-\infty,l)-l_bV_a(-\infty,l)=L_{ab}\Phi(l)\,,
\end{equation}
where $\Phi(l)$ is a homogeneous function given by
\begin{equation}\label{Fi}
 \Phi(l)=\frac{1}{4\pi}\int\frac{l\cdot V(-\infty,l')}{l\cdot l'}\,\dl'\,.
\end{equation}
Further properties of the relation of $V$ with $\Phi$, which will be used later, are to be found in Appendix \ref{unchv}.

We end this section with a simple illustration of the null structures for the potential of the scattered particle \eqref{qrad}. In this case $V(s,l)=\theta(-s)V(-\infty,l)$ ($\theta$ is the Heaviside step function; the singularity at $s=0$ is again due to unrealistic instantaneous scattering), $V(-\infty,l)=q[(v/v\cdot l)-(u/u\cdot l)]$ and  $\Phi(l)=q\log(v\cdot l/u\cdot l)$.

\section{Matter-radiation system: null and spacelike infinity}\label{nullspacelike}

Let now the inhomogeneous Maxwell equations be a part of a closed electrodynamic system. We again assume Lorenz gauge, thus
\begin{equation}
 \Box A(x)=4\pi J(x)
\end{equation}
in Gaussian units, which we use in this paper.\footnote{Let us remark in passing that the authors of \cite{kps15} write the inhomogeneous Maxwell equation with $e^2J$ on the rhs. This implies a possible, but a rather exotic units system.} We consider a scattering situation, thus the conserved current $J(x)$ describes sources which stabilize in remote past and future. This stabilization, more specifically,  means that for $v$ on the future unit hyperboloid the leading behavior of the current is
\begin{equation}\label{asJ}
 J(\pm\la v)\sim \la^{-3}v\,\rho_\pm(v)\quad \text{for}\quad \la\to\infty\,,
\end{equation}
with some scalar functions $\rho_\pm(v)$ on the hyperboloid. This is precisely true for asymptotically freely moving classical particles, but also, up to oscillating terms not contributing to the leading behavior of $A$, for matter fields like massive  Dirac field. Also, we assume that $J(x)$ vanishes sufficiently fast in spacelike directions. It follows then that the function defined by
\begin{equation}
 V_J(s,l)=\int J(x)\,\delta(s-l\cdot x)\,d^4x\,,
\end{equation}
where $\delta$ is the Dirac measure, is well-defined, bounded and for $s$ tending to $\pm\infty$ has finite limits
\begin{equation}\label{VJinfty}
 V_J(\pm\infty,l)=\int\frac{v}{v\cdot l}\,\rho_\pm(v)\,d\mu(v)\,,
\end{equation}
the integral over the unit hyperboloid. These limits are fully determined by the asymptotic forms of the current in the future/past. The conserved charge carried by the current is
\begin{equation}
 Q=l\cdot V_J(s,l)= \int\rho_\pm(v)\,d\mu(v)\,.
\end{equation}
Moreover, due to \eqref{asJ} one has
\begin{equation}
 L_{[ab}V_J{}_{c]}(\pm\infty,l)=0\,.
\end{equation}
The reason to define this function is that the null asymptotes of the retarded and advanced fields are then given by
\begin{gather}
 \lim_{R\to\infty}RA^\ret(x+Rl)=V_J(x\cdot l,l)\,,\quad \lim_{R\to\infty}RA^\ret(x-Rl)=V_J(-\infty,l)\,,\\
 \lim_{R\to\infty}RA^\adv(x-Rl)=V_J(x\cdot l,l)\,,\quad \lim_{R\to\infty}RA^\adv(x+Rl)=V_J(+\infty,l)\,.
\end{gather}
It also follows that the radiation field $A^\rad=A^\ret-A^\adv$ is the free field given by formula \eqref{freedl} in which one should substitute $\V_J$ for $\V$, and it satisfies all our selection criteria.

The total potential is decomposed in two ways appropriate for incoming and outgoing characterization, $A=A^\ret+A^\inc=A^\adv+A^\out$, respectively, where $A^\inc$ and $A^\out$ are free fields. It is now evident that if the incoming field $A^\inc$ satisfies our selection criteria, then also the outgoing field $A^\out$ does. We again denote by $V^\out$ and $V^\inc$ the future asymptotes, and by ${V^\out}'$ and ${V^\inc}'$ the past asymptotes of $A^\out$ and $A^\inc$, respectively. Then it follows that the null asymptotes of the total potential and the total field are given by
\begin{equation}\label{nullA}
 \lim_{R\to\infty}RA(x+Rl)=V(x\cdot l,l)\,,\quad \lim_{R\to\infty}RA(x-Rl)=V'(x\cdot l,l)\,,
\end{equation}
\begin{equation}\label{nullF}
\begin{split}
  &\lim_{R\to\infty}RF_{ab}(x+Rl)=l_a\V_b(x\cdot l,l)-l_b\V_a(x\cdot l,l)\,,\\[1ex]
  &\lim_{R\to\infty}RF_{ab}(x-Rl)=l_a\V'_b(x\cdot l,l)-l_b\V'_a(x\cdot l,l)\,,
\end{split}
\end{equation}
where
\begin{equation}\label{VV'}
 \begin{gathered}V(s,l)=V_J(s,l)+V^\inc(s,l)=V_J(+\infty,l)+V^\out(s,l)\,,\\[1ex]
 V'(s,l)=V_J(-\infty,l)+V^\inc{}'(s,l)=V_J(s,l)+V^\out{}'(s,l)\,.
 \end{gathered}
\end{equation}
Putting here $s=\pm\infty$ one finds that
\begin{equation}
 V(+\infty,l)=V_J(+\infty,l)\,,\quad V'(-\infty,l)=V_J(-\infty,l)\,,
\end{equation}
so these characteristics are totally due to the outgoing/incoming currents, respectively, while
\begin{equation}
\begin{aligned}
 &V(-\infty,l)=V_J(+\infty,l)+V^\out(-\infty,l)\,,\\[1ex]
 &V'(+\infty,l)=V_J(-\infty,l)+V^\inc{}'(+\infty,l)\,,
\end{aligned}
\end{equation}
so these characteristics are the sums of the source and the free infrared-singular contributions. Moreover, adding the sides of equations \eqref{VV'} one finds that $V(s,l)+V'(s,l)-V_J(s,l)$ is constant in $s$, thus
\begin{equation}\label{match}
 V(s,l)+V'(s,l)-V_J(s,l)=V'(+\infty,l)=V(-\infty,l)\,.
\end{equation}
The second equality in this formula constitutes a matching property of the values at the past edge of the future null infinity, and the future edge of the past null infinity (see \cite{he95}, Eq.\ (2.26) and the following discussion). We stress once more that to this matching equality \emph{contribute both retarded/advanced as well as free `in'/`out' fields}.\footnote{It is not true that ``finite energy wave packets die off at $i^0$'' as the authors of \cite{kps15} put it.} The spacelike asymptotic behavior is also governed by this characteristic, for $y^2<0$:
\begin{align}
 \lim_{R\to\infty}RA_b(x+Ry)&=\int V_b(-\infty,l)\delta(y\cdot l)\dl\,,\label{spA}\\
 \lim_{R\to\infty}R^2F_{ab}(x+Ry)&=\int \big(l_aV_b(-\infty,l)-l_bV_a(-\infty,l)\big)\delta'(y\cdot l)\dl\,,\label{spF}
\end{align}
where $\delta'$ is the derivative of the Dirac delta distribution.\pagebreak

At this point an observation on gauge of the total potential is due. As we work in a scheme which assumes Lorenz gauge, the gauge transformation $A\rightarrow A+\p\Lambda$ demands $\Box\Lambda=0$. Therefore, $\p\Lambda$ is a free potential satisfying all our selection criteria. In particular, if $l\alpha(s,l)$ is its future null asymptote, then $\alpha(+\infty,l)=0$, as in the free case with no sources (and similarly $\alpha'(-\infty,l)=0$ for past asymptote). Thus the limit values $V(+\infty,l)$ and $V'(-\infty,l)$ are gauge-independent (in the Lorenz-gauge class).\footnote{In this connection we note that there is an erroneous remark on `constant gauge' transformation in \cite{he98} preceding formula (3.18), but in fact no such freedom exists and this formula gives the correct unique value of $V(+\infty,l)$. This lapse has no further continuation or consequences in that article.}

We end this section with an illustration of the `matching property' \eqref{match}. Consider a scattering event in which there are $n'$ incoming particles and $n$ outgoing particles, with charges and four-velocities given by $q'_i, v'_i$ and $q_j, v_j$, respectively. It is easy to show that for a single free particle with charge $q$ and velocity $v$ one has $V_J(s,l)=q\,v/v\cdot l$. Therefore, in our scattering event we have
\begin{equation}
 V_J(-\infty,l)=\sum_{i=1}^{n'}q'_i\,\frac{v'_i}{v'\cdot l}\,,\quad
 V_J(+\infty,l)=\sum_{i=1}^{n}q_i\,\frac{v_i}{v\cdot l}\,.
\end{equation}
Then Eq.\ \eqref{match} gives
\begin{equation}\label{infra}
 2\pi \lim_{\w\to0}\w \big(a^\out(\w l)-a^\inc(\w l)\big)=\sum_{i=1}^{n'}q'_i\frac{v'_i}{v'\cdot l} - \sum_{i=1}^{n}q_i\frac{v_i}{v\cdot l}\,.
\end{equation}
This relation is a clear announcement of the problems to come in quantum theory: Suppose the space of incoming states does not contain infrared-singular `in' fields, so the `in' fields are represented in the usual Fock space of photons. However, this is only possible if the profiles of incoming photon states satisfy $\dsp\lim_{\w\to0}\w a^\inc(\w l)=0$. On the other hand, the rhs of Eq.\,\eqref{infra} cannot vanish for all scattering amplitudes (depending on velocities $v'_i$ and $v_i$) in the given representation. The relation tells us then that $\dsp\lim_{\w\to0}\w a^\out(\w l)\neq0$, which contradicts the possibility of representing the `out' profiles in Fock space. We shall come later to existing strategies for the resolution of this difficulty.

\section{Matter-radiation system: timelike infinity}\label{timelike}

We start by recalling the Fourier representation of the solution of the free Dirac equation. A convenient version of this integral representation is:
\begin{equation}\label{DF}
 \psi(x)=\Big(\frac{m}{2\pi}\Big)^{3/2} \int e^{{\textstyle-im\,x \cdot v\,\ga\cdot v}}\ga\cdot v\, f(v)\,d\mu(v)\,,
\end{equation}
where $d\mu(v)=d^3v/v^0$ is the standard invariant measure on the unit future hyperboloid and $f(v)$ is a Dirac spinor-valued function on this hyperboloid. A more standard-looking form is obtained by noting that
\begin{equation}
 e^{{\textstyle-im\,x \cdot v\,\ga\cdot v}}\ga\cdot v=e^{-im\,x \cdot v}P_+(v)-e^{+im\,x \cdot v}P_-(v)\,,
\end{equation}
where $P_\pm(v)=\frac{1}{2}(1\pm \ga\cdot v)$.
While the free electromagnetic field is fully encoded in its null asymptote, the free Dirac field is fully encoded in its timelike asymptote, which is
\begin{equation}\label{Das}
 \psi(\pm\la v)\, \sim\, \mp i\,\la^{-3/2} e^{{\textstyle \mp i( m\la + \pi /4)\ga\cdot v}} f(v)\qquad \text{for}\quad \la\to\infty\,.
\end{equation}
Remember that the free field has no gauge freedom of the second kind, so $\psi$ and $f$ are unique up to a constant phase.

Suppose now that the current $J$ appearing in the last section is produced by the Dirac electron-positron field, and the fields form a closed theory with the Dirac equation added. It is well-known that the full Dirac field in standard Lorenz gauge does not approach then in simple manner any free field for time tending to plus/minus infinity, its  leading asymptotic term containing electromagnetic contributions. The standard technique usually employed for handling this problem was proposed by Dollard \cite{do64}, and in quantum electrodynamics was applied by Kulish and Faddeev \cite{kf70}. However, we find it more convenient, for reasons to become clear below, to use another procedure.

We introduce a new gauge of the total electromagnetic potential: if $A$ is the Lorenz potential discussed in the preceding sections, then
\begin{equation}\label{Atr}
 A_\tr(x)=A(x)-\p S(x)\,,
\end{equation}
where $S$ is any scalar function assumed to satisfy
\begin{equation}\label{gauge}
 S(x)\simeq \log\sqrt{x^2}\,x\cdot A(x) \quad\text{for}\quad x^2\to\infty\,.
\end{equation}
For $x=\pm\la v$ the leading asymptotic terms of $A(\pm\la v)$ for $\la\to\infty$ are of order $1/\la$. It follows then that in this limit $v\cdot A_\tr(\pm\la v)=O(\la^{-1-\ep})$. Now we can use one of the results of \cite{he95}, which says the following. \emph{In the gauge defined above the full Dirac field again has the asymptotic behavior given by \eqref{Das}, with some spinor functions~$f_\mp$ for `in' and `out' asymptotes respectively.} There are no electromagnetic field modifications of this asymptotic terms, and one could define free incoming/outgoing fields by plugging $f_\mp$ into the integral representation \eqref{DF}. We indeed intend to perform such operation, but with one additional modification explained in the next section.

We note an important fact that there is no gauge freedom of the second kind in the definition of asymptotes $f_\mp$. Indeed, one can show that $f_\mp$ do not depend either on the choice of particular Lorenz gauge of $A$, or of the specific function $S$ satisfying \eqref{gauge}.

\section{Invariant structures and conserved quantities}\label{invstruct}

The causal `in' and `out' infinities are equipped with two Poincar\'e invariant structures: symplectic form at the null infinity and scalar product in the timelike infinity. In the following we specify to the `out' infinity, similar structures exist in the `in' case. Namely, the Poincar\'e transformations acting by
\begin{equation}
\begin{aligned}
&[T_{x,A} V]_a (s,l)= \Lambda(A)_a{}^b\, V_b(s-x\cdot l,\Lambda^{-1} l)\,,\\
&[R_{x,A}f](v) = e^{\txt imx\cdot v\ga\cdot v}S(A)f(\Lambda^{-1}v)
\end{aligned}
\end{equation}
leave invariant the symplectic form
\begin{equation}\label{symf}
 \{V_1,V_2\}= \frac{1}{4\pi}\int(\V_1\cdot V_2-\V_2\cdot V_1)(s,l)\,ds\,\dl=\{T_{x,A}V_1,T_{x,A}V_2\}\,,
\end{equation}
and the pre-Hilbert scalar product
\begin{equation}\label{scpr}
 (f_1,f_2)=\int \ov{f_1(v)}\ga\cdot v f_2(v)\,d\mu(v)=(R_{x,A}f_1,R_{x,A}f_2)\,.
\end{equation}
Here $(x,A)$ is an element of the affine extension by translations of the group $SL(2,\mC)$, $\Lambda(A)$ is the corresponding Lorentz transformation, and $S(A)$ the bispinor representation.

For two solutions $\psi_i$ ($i=1,2$) of the Dirac equation in electromagnetic field one has the scalar product defined by the integral over any Cauchy surface: $(\psi_1,\psi_2)=\int \ov{\psi(x)}\gamma^a\psi(x)\,d\sigma_a(x)$. The product \eqref{scpr} expresses this quantity in terms of the asymptotic variables defined at the end of Section~\ref{timelike}.

For a pair of free, infrared regular electromagnetic fields one has the standard symplectic form defined by the integral over any Cauchy surface by $\{F_1,F_2\}=\frac{1}{4\pi}\int(F_1^{ab}A_{2b}-F_2^{ab}A_{1b})(x)\,d\sigma_a(x)$. This form, when expressed in terms of null asymptotes is equal to the form \eqref{symf}. As can be easily seen, the integral \eqref{symf} extends to all free fields in the class considered in this article (it is absolutely convergent also for infrared singular fields\footnote{We note that, in contrast, the integral over a Cauchy surface is not absolutely convergent for such fields, and one needs a regularization to make it well-defined, see \cite{he08}.}). Note that for free fields the form \eqref{symf} is gauge-invariant: it does not change under the transformation $V_i(s,l)\mapsto V_i(s,l)+l\alpha_i(s,l)$ ($i=1,2$).

The form \eqref{symf} extends also to the asymptotes of total fields, with the integral staying absolutely convergent. In this case the gauge transformation causes the following change of this quantity:
\begin{equation}\label{gauch}
 \{V_1+l\alpha_1,V_2+l\alpha_2\}=\{V_1,V_2\}+\frac{Q_1}{4\pi}\int\alpha_2(-\infty,l)\dl-\frac{Q_2}{4\pi}\int\alpha_1(-\infty,l)\dl\,,
\end{equation}
where $Q_i=l\cdot V_i(s,l)$ is the charge of the field $V_i$. This addition will play a~role in the quantum case.

Let us return to the Poincar\'e transformations of the asymptotic structures. The generators of these transformations defined by
\begin{equation}
T_{x,A}-1\approx x^ar_a-\frac{1}{2}\w^{ab}n_{ab}\,,\quad
R_{x,A}-1\approx ix^ap_a-\frac{i}{2}\w^{ab}m_{ab}\,,
\end{equation}
for infinitesimal $x^a$ and $\w^{ab}$,
where $\Lambda^a{}_b\approx g^a_b+\w^a{}_b$, are
\begin{gather*}
 (r_aV)_c(s,l)=-l_a\V_c(s,l)\,,\quad  (p_af)(v)=mv_a\ga\cdot vf(v)\,,\\[1ex]
 (n_{ab}V)_c(s,l)=-L_{ab}V_c(s,l)-g_{ca}V_b(s,l)+ g_{cb}V_a(s,l)\,,\\[1ex]
 (m_{ab}f)(v)=i\big(\mu_{ab} +\tfrac{1}{4}[\ga_a,\ga_b]\big) f(v)\,,
\end{gather*}
where $\mu_{ab}=v_a(\p/\p v^b)-v_b(\p/\p v^a)$ acts intrinsically in the hyperboloid. Consider now the energy-momentum and $4$-angular momentum going out into timelike and null infinity, carried respectively by massive and electromagnetic fields.
One finds that these quantities are elegantly expressed by
\begin{equation}\label{PMf}
\begin{aligned}
P^\out_a&=(f_+, p_af_+) +{\txt\frac{1}{2}}\{ V,r_aV\}\,,\\[1ex]
M^\out_{ab}&=(f_+, m_{ab}f_+) +{\txt\frac{1}{2}}\{ V,n_{ab}V\}\,,
\end{aligned}
\end{equation}
where $V(s,l)$ is the asymptote of the total electromagnetic potential, and $f_+(v)$ is the asymptotic variable of the Dirac field as defined at the end of Section \ref{timelike}. All four expressions on the rhs of these equations are \emph{gauge-invariant}: the scalar product parts because of the second kind gauge-invariance of $f_+$, and the invariance of the symplectic form parts is proved by an easy calculation using the form of generators and the identity \eqref{gauch}.

If analogous incoming quantities $P^\inc_a$ and $M^\inc_{ab}$ are formed, then one finds the expected conservation laws $P^\inc=P^\out$ and $M^\inc=M^\out$. In case of the energy-momentum this is also equal to the integral of the total energy-momentum density over any Cauchy surface, but in case of the angular momentum a similar statement needs a comment. The $1/R^2$ spacelike tail of the electromagnetic field has the consequence that the density of the angular momentum is not absolutely integrable over a Cauchy surface. However, the oddness of the tail of electromagnetic field implies that this density is also asymptotically odd, which allows one to regularize the integral over spacelike hyperplane by taking a limit of integrals over balls of finite radius. The quantity so regularized is equal to the incoming and the outgoing angular momentum. As an aside we note the following interesting fact. In a theory in which both electric and magnetic charges are present, the analogous statements are not true; namely $M^\inc\neq M^\out$ -- the angular momentum leaks into the spacelike infinity \cite{he95}.

If we now decompose $V(s,l)=V(+\infty,l)+V^\out(s,l)$, then $\{ V,r_aV\}=\{ V^\out,r_aV^\out\}$ -- the electromagnetic energy-momentum is expressed in terms of the free `out' field. However, the angular momentum turns out to contain a~mixed adv-out term:
\begin{equation}
 \tfrac{1}{2}\{ V,n_{ab}V\}=\tfrac{1}{2}\{ V^{{\rm out}},n_{ab}V^{{\rm out}}\} + \{ V^{{\rm out}},n_{ab}V(+\infty,.)\}\,.
\end{equation}
Using the asymptotic form of the Dirac field one finds
\begin{equation}\label{Vplus}
 V_a(+\infty,l)=\int n(v) V^e_a(v,l)\, d\mu(v)\,,
\end{equation}
where $n(v)=\ov{f(v)}\ga\cdot v f(v)$ is the asymptotic density of
particles moving with velocity $v$ and $V_a^e(v,l) =ev_a/v\cdot l$ is the null asymptote of the Lorentz potential of the Coulomb field surrounding the particle with charge $e$ moving with constant velocity $v$. Noting the identity $(n_{ab}V^e)_c(v,l)=\mu_{ab}V_c^e(v,l)$, we obtain
\begin{equation}
 \{ V^{{\rm out}},n_{ab}V(+\infty,.)\}=-\int n(v)\mu_{ab}\{V^e(v,.),V^\out\}d\mu(v)\,.
\end{equation}
The latter form of this term allows its absorption into the matter part by a~phase transformation. We introduce
\begin{equation}\label{g}
 g(v)=\exp\big(i\{V^e(v,.),V^\out\}\big)f(v)=\exp\big(i\{V^e(v,.),V\}\big)f(v)
\end{equation}
(the second form following from $\{V^e(v,.),V(+\infty,.)\}=0$) and then the conserved quantities take the form
\begin{equation}\label{PMg}
\begin{aligned}
P^\out_a&=(g_+, p_ag_+) +{\txt\frac{1}{2}}\{ V^\out,r_aV^\out\}\, ,\\
M^\out_{ab}&=(g_+, m_{ab}g_+) +{\txt\frac{1}{2}}\{ V^\out,n_{ab}V^\out\}\, .
\end{aligned}
\end{equation}
These expressions look formally as sums of independent free fields contributions. Therefore, we identify the Dirac out-field by placing $g_+$ in place of $f$ in the integral representation \eqref{DF}. However, this Dirac field must then describe massive particles together with their Coulomb fields. This will be confirmed in quantization. We end this section by giving another form to the $f\to g$ transformation: if one uses the definition \eqref{Fi}, then the identity \eqref{FiV} applied to the pair $V^\out(-\infty,l)$, $\Phi^\out(l)$ implies
\begin{equation}\label{gFi}
 g(v)=\exp\bigg(\frac{ie}{4\pi}\int\frac{\Phi^\out(l)}{(v\cdot l)^2}\, \dl\bigg)f(v)\,.
\end{equation}

\section{Other gauges. LGT is not a gauge symmetry of the asymptotic structure}\label{lgtrans}

We have thus completed the discussion of the asymptotic structure of the classical electrodynamics. With this picture before our eyes let us now once more return to the question of the gauge symmetry of this structure.

In our constructions we have assumed the class of the Lorenz gauges. This class is distinguished by a well-defined dynamics, and allows for gauge transformations within the class, which are encoded in the null asymptotes by $V_2(s,l)=V_1(s,l)+l\alpha(s,l)$. The question we want to discuss now is the following: are there other/wider classes of gauges for which this structure is preserved?

The class of electromagnetic fields $F_{ab}(x)$ considered here is characterized by the existence of the Lorenz gauge potentials $A_b(x)$ with well defined null asymptotes $V_b(s,l)$. We now want to admit other gauges $\hat{A}=A+\p\Lambda$, but we keep the restriction demanding that the potentials have well defined null asymptotes, i.e.\ for each $\hat{A}_a(x)$ there is
\begin{equation}
 \lim_{R\to\infty}R\hat{A}_b\Big(\frac{st}{t\cdot l}+Rl\Big)=\hat{V}_b(s,l)
\end{equation}
---the denominator $t\cdot l$ guarantees that the asymptote $\hat{V}_b(s,l)$ has the same homogeneity property as $V_b(s,l)$: $\hat{V}(\lambda s,\lambda l)=\lambda^{-1}\hat{V}(s,l)$, $\lambda>0$ (see the discussion of the variables $R,s,l$ in Appendix \ref{Rsl}).
It follows that for the gauge function $\Lambda(x)$ the limit
\begin{equation}
 \lim_{R\to\infty}R(\p_b\Lambda)\Big(\frac{st}{t\cdot l}+Rl\Big)=\hat{V}_b(s,l)-V_b(s,l)
\end{equation}
has to exist. This condition restricts the gauge functions to those with the asymptotic expansion
\begin{equation}\label{expla}
 \Lambda\Big(\frac{st}{t\cdot l}+Rl\Big)=\vep^+(l)+\frac{\beta_t(s,l)}{R}+o(R^{-1})\,,
\end{equation}
with the sufficiently regular rest, so that the leading term of $\p\Lambda$ is obtained by differentiation of the two leading explicitly written terms; note that $\vep^+(l)$ is homogeneous of degree $0$, and $\beta_t(s,l)$ is homogeneous of degree $-1$.

Let $V^+_b(l)$ be a vector function entering with $\vep^+(l)$ into the relation described in Appendix \ref{unchv} for the pair $V$ and $\Phi_V$, formulae \eqref{Vunch}--\eqref{FiV}. Using these relations, and applying the formula for differentiation, Eq.\,\eqref{difsl} in Appendix~\ref{Rsl}, we obtain
\begin{multline}\label{expga}
 (\p_b\Lambda)\Big(\frac{st}{t\cdot l}+Rl\Big)=\frac{1}{R}\Big[\frac{t^a}{t\cdot l}L_{ab}\vep^+(l)+l_b\dot{\beta}_t(s,l)\Big]+o(R^{-1})\\
 =\frac{1}{R}\Big[V^+_b(l)+l_b\Big(\dot{\beta}_t(s,l)-\frac{t\cdot V^+}{t\cdot l}\Big)\Big]+o(R^{-1})\,.
\end{multline}

The form of the expansions \eqref{expla} and \eqref{expga} remains valid if the time axis vector $t$ is replaced by another such vector $t'$. However, while the function $\vep^+(l)$ remains unchanged (as the notation suggests), the function $\beta_t$ undergoes the following transformation:
\begin{equation}\label{transbet}
 \beta_{t'}(s,l)-\beta_t(s,l)=\frac{s}{t\cdot l\, t'\cdot l}\,t^a{t'}^bL_{ab}\vep^+(l)=s\frac{t'\cdot V^+}{t'\cdot l}-s\frac{t\cdot V^+}{t\cdot l}
\end{equation}
-- this is shown by writing $(s'/t'\cdot l')+R'l'=(s/t\cdot l)+Rl$ and expressing $R',s',l'$ in terms of $R,s,l$ in leading orders in $1/R$.
This shows that $\beta_t(s,l)-st\cdot V^+(l)/t\cdot l=\gamma(s,l)$ does not depend on $t$, and with this notation the most general transformation of the null asymptote of the potential is given by $\hat{V}_b(s,l)=V_b(s,l)+V^+_b(l)+l_b\dot{\gamma}(s,l)$. However, if $|\ddot{\gamma}(s,l)|=O(|s|^{-1-\ep})$ for $|s|\to\infty$, which is needed for $\hat{V}(s,l)$ to satisfy the demands of the asymptotic structure, then the term $l_b\dot{\gamma}(s,l)$ is just the gauge transformation previously considered in the Lorenz class,\footnote{Up to a possible term $l_b\dot{\gamma}(+\infty,l)$, which may be included into the term $V^+_b(l)$ by an appropriate addition of a constant to $\vep^+(l)$.} so this addition brings nothing new. Therefore, we set $\gamma(s,l)=0$ and we are left with the transformation
\begin{equation}\label{lgt}
 \hat{V}_b(s,l)=V_b(s,l)+V^+_b(l)\,.
\end{equation}
This is, in a somewhat different guise, the `large gauge transformation' considered by the authors of \cite{kps15}, \cite{cl15}, and in numerous other recent articles.

We now look at the transformation \eqref{lgt} from the intrinsic point of view of the asymptotic structure as described in preceding sections. Does the new asymptote $\hat{V}(s,l)$ fit into the picture? It does, but with interpretation wholly different from the gauge transformation which led to the relation \eqref{lgt}.

First of all, the addition $V^+(l)$ to the asymptote of the potential, if it results from the gauge transformation as discussed above,  does not have the meaning of the initial data at infinity as before, as no equation is given which propagates this data into the bulk of the Minkowski space. Any rigid choice of the bulk functions is rather a transformation into a different sector, than a~symmetry.

Secondly, and even more importantly, the transformation \eqref{lgt} is not a~symmetry of the structure. The symplectic structure is changed under this transformation: if $\hat{V}_i(s,l)=V_i(s,l)+V^+_i(l)$, $i=1,2$, then
\begin{equation}
 \{\hat{V}_1,\hat{V}_2\}=\{V_1,V_2\}+\frac{1}{4\pi}\int \big[V^+_1(l)\cdot \Delta V_2(l)-V^+_2(l)\cdot \Delta V_1(l)\big]\,d^2l\,,
\end{equation}
where $\Delta V_i(l)=V_i(-\infty,l)-V_i(+\infty,l)$. This addition to the symplectic form does not vanish even in the free field case (characterized by  $V_i(+\infty,l)=0$). Moreover, the electromagnetic field contribution to the angular momentum \eqref{PMf} changes:
\begin{equation}
 \tfrac{1}{2}\{\hat{V},n_{ab}\hat{V}\}=\tfrac{1}{2}\{V,n_{ab}V\}+\{V,n_{ab}V^+\}\,.
\end{equation}

On the other hand, the transformation \eqref{lgt} \emph{may} be interpreted within the asymptotic structure of the preceding sections, but it then \emph{changes the physics} of the system. Namely, if we decompose $V(s,l)$ into the sum of the advanced part of the outgoing current
\begin{equation}
 V(+\infty,l)=\int\frac{v}{v\cdot l}\,\rho_+(v)\,d\mu(v)
\end{equation}
(see the discussion in Section \ref{nullspacelike}) and the free outgoing field $V^{\out}(s,l)=V(s,l)-V(+\infty,l)$, then we see that $\hat{V}(+\infty,l)=V(+\infty,l)+V^+(l)$ and $\hat{V}(s,l)-\hat{V}(+\infty,l)=V^\out(s,l)$. The transformation \eqref{lgt} does not change the asymptote of the potential of the free `out' field. On the other hand, we show in Appendix \ref{unchv} that for almost all $V^+$ (the precise sense of this statement is to be found there) there exists the representation
\begin{equation}
 V^+(l)=\int\frac{v}{v\cdot l}\,\rho(v)\,d\mu(v)\,,
\end{equation}
where $\rho(v)$ is a smooth function of compact support on the unit hyperboloid, and such that $\int \rho(v)\,d\mu(v)=0$. This uniquely induces the interpretation of the transformation \eqref{lgt} within the limits of the asymptotic structure:
\begin{quote}
\emph{The transformation \eqref{lgt} adds the advanced electromagnetic field of an external, asymptotic, charge-free current $J^+(x)=x\rho(x)$, $\rho(x)$ homogeneous of degree $-4$}.
\end{quote}
The current $J^+$ is not unique: the function $\rho(v)$ depends on three real variables, while the asymptote $V^+(l)$ is a function of two independent real variables.

As we shall see, the above interpretation will also be confirmed in the quantum case.

\section{Electromagnetic memory}\label{memory}

The long-range electromagnetic variables play a substantial role, both in the classical structure described above, as well as in the quantum theory to be developed below. We stop here to ask what is the observational status of such degrees of freedom.

In this connection it is appropriate to invoke certain analogy between electromagnetic and gravitational waves. Within the limits of linearized gravity, it has been pointed out in 1974 by Zel'dovitch and Polnarev \cite{zp74} that if a~pair of test bodies is exposed to a gravitational wave burst, then their relative separation undergoes a finite and permanent shift between remote past, when the burst has not arrived yet, and remote future, when it has passed away. This effect was discussed again and named the `memory' of the gravitational wave-burst in 1985 \cite{bg85}. The size of this effect depends on the initial and final masses and velocities of the bodies which produce gravitational wave-burst (see, e.g. \cite{th92}). In 1991 it was shown by Christodoulou \cite{ch91} that a substantial, and potentially more prominent, contribution to the displacement of the test bodies is caused by nonlinear effects of gravitation.

In the meantime, in 1981, Staruszkiewicz has discussed an electromagnetic effect which may also be interpreted as a kind of memory of wave-packets \cite{st81}. To my best knowledge this was the first observation of this kind in electrodynamics. Moreover, Staruszkiewicz's effect shows in the most clear way the connection between `memory' and infrared degrees of freedom. We describe the effect in our notation. Let $A(x)$ be a free Lorenz potential with the asymptote $V(s,l)$, thus given by the formula \eqref{freedl}. The scaled potential $A_\la(x)=\la^{-1}A(x/\la)$, $\la>0$, with the asymptote $V_\la(s,l)=V(s/\la,l)$, has the same spacelike tail (as it depends on $V_\la(-\infty,l)=V(-\infty,l)$, see \eqref{spA}), but with $\la\to\infty$ the energy content of the scaled electromagnetic field tends to zero; we call this large $\la$ regime the infrared limit. In this limit the field is to weak to change the momentum of a charged test particle. However, as shown by Staruszkiewicz in the quasi-classical approximation, the phase of the wave function of the particle undergoes, between times minus and plus infinity, the momentum-dependent phase shift
\begin{equation}
 \delta(p)=-\frac{q}{2\pi}\int\frac{p\cdot V(-\infty,l)}{p\cdot l}\,\dl\,,
\end{equation}
where $q$ is the charge of the particle and $p$ its momentum. For particle wave-packets this produces interference effects, which result in an adiabatic finite shift of the wave-packet of the particle. The shift is solely due to the long-range degrees of freedom of the electromagnetic field; the field itself may have a~negligible energy content. Later this effect was confirmed in more elevated settings in \cite{he95} and \cite{he12} (in the latter reference a quantum Dirac field in the external infrared limit field was considered).\footnote{There is, recently, a considerable activity in the field of `electromagnetic memory' and its connection with the long-range structure of electrodynamics, apparently not aware of Staruszkiewicz's work, and the articles just mentioned; see \cite{pa15} and the literature cited there.}

The Staruszkiewicz effect may be simply explained for a classical test particle with charge $q$ and mass $m$ placed in the free external electromagnetic field $F_{ab}(x)$. The equation of motion for the trajectory $z^a(\tau)$:
\begin{equation}
 \ddot{z}^a(\tau)=\tfrac{q}{m}F^a{}_b(z(\tau))\dot{z}^b(\tau)\,,
\end{equation}
where $\tau$ is the proper time and the overdot denotes differentiation with respect to $\tau$, may be written in the equivalent form
\begin{equation}
 z^a(\tau)=z^a(0)+\tau\dot{z}^a(\tau)-\frac{q}{m}\int_0^\tau F^a{}_b(z(\sigma))\dot{z}^b(\sigma)\,\sigma\,d\sigma\,.
\end{equation}
For $\tau$ tending to $\pm\infty$ the integrals $\Delta_\pm^a=-(q/m)\int_0^{\pm\infty}F^a{}_b(z(\sigma))\dot{z}^b(\sigma)\,\sigma d\sigma$ give the vectors separating the `out'/`in' asymptotic free trajectories (straight lines) from the point $z^a(0)$. For scaled $F$, in the infrared limit we have $\dot{z}(\tau)\simeq v$, a~constant four-velocity, and then
the `out' asymptotic trajectory is parallel to the `in' asymptotic trajectory, but shifted by $\Delta^a=\Delta_+^a-\Delta_-^a$. Substituting in these integrals $z(\sigma)\simeq z(0)+v\sigma$, using the representation \eqref{freeFdl}, and integrating by parts in $\sigma$ we find in the infrared limit\footnote{We assume that the field $F$ satisfies all our selection criteria, in particular the condition~\eqref{ima}.}
\begin{equation}
 \Delta^a(v)=-\frac{q}{m}\int\limits_{-\infty}^{+\infty}F^a{}_b(v\sigma)v^b\sigma d\sigma=\frac{q}{2\pi m}\int \frac{l^aV^b(-\infty,l)-l^bV^a(-\infty,l)}{(v\cdot l)^2}\,\dl\,v_b\,.
 \end{equation}
Now we use the relation \eqref{VFi}, followed by \eqref{Ld2l}, to obtain
\begin{equation}
 \Delta^a(v)=\frac{q}{\pi m}\int (l^a-l\cdot v v^a)\frac{\Phi(l)}{(v\cdot l)^3}\,\dl\,.
\end{equation}
All steps in this sketchy derivation could be made watertight, what we are not going to discuss here.

\section{Quantum theory}\label{quantum}

Construction of interacting quantum field theories faces numerous problems. Not only perturbative solutions of nonlinear equations of physically realistic models involving quantum fields need elaborate techniques in order to avoid the well-known ultraviolet divergencies, but the very existence of these equations is not clear beyond perturbative formulation. Whether this state of art is due only to technical problems, or is some better physical input needed, in our opinion is still not clear. However, on the perturbative level the Epstein-Glaser technique \cite{eg73, sc95} of proper time-splitting of distributions supplies the most fundamental, conceptually clear, even if calculationally formidable, solution of the problem.

However, here we are interested rather in the opposite, in terms of momentum transfer, end of scale, the infrared problems. In perturbative calculations these difficulties manifest themselves in the appearance of divergencies in Feynmann diagrams for small momenta. But this perspective only touches the tip of an iceberg: in this case the source of the difficulties is definitely not merely technical, but related to the long-range character of the interaction. To understand the source of the problem one has to consider the theory from its algebraic structure and representation perspective.

One of the paradigms on which the standard quantum field theory is build is relativistic locality. This is best expressed in algebraic terms \cite{ha92}, and consists of two assumptions: (i) basic quantum observables are localized in compact spacetime regions, and all other observables are limits of such quantities; (ii) observables localized in regions spacelike separated commute. For electrodynamics this has the following consequence \cite{gz80, bu82, bu86}: Suppose that in the representation space one can define the spacelike asymptotic field of the type indicated by the classical expression \eqref{spF}. As the localization of this field becomes spacelike to any compact set in the limit, this asymptotic field commutes with all observables. In any irreducible representation this implies that the field is proportional to the unit operator and its specific value is a superselection label of the representation. In the language of previous sections this means that in such setting $V(-\infty,l)=V'(+\infty,l)$ is a classical function characterizing the choice of representation. In view of this knowledge we look once more at relations \eqref{match} and \eqref{infra}. If quantum charged free `in' and `out' particles are ascribed their Coulomb fields, then incoming and outgoing free electromagnetic fields cannot be both represented in unitarily equivalent way (e.g., as mentioned before, if $a^\inc$ is Fock, then $a^\out$ cannot be Fock). Thus the description of scattering falls into troubles.

Standard way to deal with this problem is to `dress' charged particles, in addition to their Coulomb fields, with some permanent `radiation clouds'. These clouds must be chosen so as to decouple the long-range tail of the field attached to the charged particle from its specific momentum and make the rhs of the relation \eqref{infra} vanish.\footnote{Some recent examples of dressing constructions include \cite{cfp09, cfp10, ms14}, where also further bibliographic information may be found.}  Another potentially possible strategy is to use representations which do not satisfy the above assumptions and are so singular in the infrared/long range regime as to be able to mask radiative soft additions (e.g.\ so called KPR infravacua \cite{kpr77}).

In what follows we want to avoid any of the above two paths. We insist that free charged particles carry their Coulomb fields and these fields alone form their structural parts. The only way not to come in conflict with the above discussion is to admit certain degree of nonlocality in the theory: the long-range variables will acquire genuinely quantum character. Having agreed to this we proceed as follows.

\section{Construction of the algebra}\label{constalg}

We want to use the heuristic method of `relativistic quantization', in which one demands that the classical expressions for energy-momentum and 4-an\-gu\-lar momentum should become the Poincar\'e generators of the quantum theory. This is not in conflict with the knowledge that Lorentz symmetry may be broken in electrodynamics: the heuristic idea is applied only on the algebraic level, while breaking of symmetry happens on the level of representations. The form of conserved quantities \eqref{PMf} suggests that $f$ and $V$ should be quantized independently. We denote the corresponding quantum variables $f^q$ and $V^q$ and then the quantization conditions are easily obtained:
\begin{gather}
\big[\{ V_1,V^{{\rm q}}\}, \{ V_2,V^{{\rm q}}\}\big] = i\{ V_1, V_2\}\,,\\[1ex]
\big[(f_1,f^{{\rm q}}), (f_2,f^{{\rm q}})\big]_+=0\,,\quad
\big[(f_1,f^{{\rm q}}), (f_2,f^{{\rm q}})^*\big]_+= (f_1, f_2)\,.
\end{gather}
Here $f_i$ and $V_i$ are test variables. The variable $V^q$ has a well- defined physical meaning as the total electromagnetic field at null infinity. On the other hand, as argued above, the physical field equipped with its Coulomb field should be defined by \eqref{g}, thus on the quantum level $g^q(v)=\exp\big(i\{V^e(v,.),V^q\}\big)f^q(v)$. We now use the pair $g^q,V^q$ as generating a closed algebra. Denoting
\begin{equation}\label{PW}
 \Psi(g_i)=(g_i,g^q)\,,\quad W(V_i)=\exp\big[-i\{V_i,V^q\}\big]
\end{equation}
one obtains
\begin{gather}
\begin{split}\label{alg1}
 &W(V_1)W(V_2)= e^{\txt -\frac{i}{2}\{ V_1,V_2\}} W(V_1+V_2)\,,\\[1ex]
 &W(V)^* = W(-V)\,,\quad W(0)=1\,,
\end{split}\\[2ex]
 [\Psi(g_1), \Psi(g_2)]_+ =0\,,\quad [\Psi(g_1), \Psi(g_2)^*]_+ = (g_1, g_2)1\,,\label{alg2}\\[1.5ex]
W(V) \Psi(g) = \Psi(S_\Phi g) W(V)\,,\label{alg3}
\end{gather}
where
\begin{equation}
\big(S_\Phi g\big)(v) = \exp\Big(\dsp i\frac{e}{4\pi} \int\frac{\Phi(l)}{(v\cdot l)^2}\,\dl\Big)\, g(v)\, .
\label{alg4}
\end{equation}
One notes that $\{V_1,V^q\}W(V_2)=W(V_2)\big(\{V_1,V^q\}+\{V_1,V_2\}\big)$. Thus $W(V_2)$, when acting on a state, adds to this state the field $V_2$. However, the Coulomb fields carried by particles are now attached to the field $\Psi(g)$, so the test fields $V_i$ should be free fields test functions, therefore  we demand \mbox{$V_i(+\infty,l)=0$}. Function $\Phi(l)$ is related to $V(-\infty,l)$ as described in Section \ref{step3}.

Relations \eqref{alg1} - \eqref{alg4} may be now considered detached from the heuristic considerations which led to them. One shows that they generate a $C^*$-algebra, which simply means that they can indeed be represented by bounded operators in a Hilbert space. Physically meaningful representations must satisfy some selection criteria, and we shall briefly comment on the form of possible representations below. But first we want to clarify the question of gauge (in)dependence of the algebra.

\section{Gauge invariance}\label{gaugeinv}

First, we recall that the test functions $V_i(s,l)$ are free field functions, in particular $l\cdot V_i(s,l)=0$. It follows that any gauge transformation of the quantum field $V^q(s,l)\to V^q(s,l)+l\al(s,l)$ leaves the quantity $\{V_i,V^q\}$ unchanged\footnote{In accordance with the remark on the gauge of potential in Section \ref{nullspacelike}, after Eq.\,\eqref{spF}, we assume $\al(+\infty,l)=0$. Moreover, we assume, as is usual in such analyses, that the gauge function $\al(s,l)$ is a c-number; otherwise the transformation of $g^q$ could pose problems and lead to a change of algebra.}. Thus $W(V_i)$ should be interpreted as exponentiations of the total electromagnetic \emph{field}, and not merely potential. Second, the field $f$ was obtained in a~gauge-independent way, and the transition from $f$ to $g$ is given by~\eqref{gFi}. Thus on the (heuristic) quantum level the gauge transformation causes the transformation $g^q\to e^{ie\sigma}g^q$, where $\sigma=\frac{1}{4\pi}\int\al(-\infty,l)\dl$. In terms of algebraic elements: $\Psi(g_i)\to e^{ie\sigma}\Psi(g_i)$. Summarizing, gauge transformations lead only to transformations of the first kind of the elements of the algebra.

Another question related to gauge symmetry is, whether the elements of the algebra depend on the gauge of test fields $V_i(s,l)$. To answer this, it is sufficient to consider $W(l\al)$ -- element with pure gauge test field. Then $\{l\alpha,V^q\}=-\sigma Q^q$, where $Q^q=l\cdot V^q$ should have the interpretation of the total charge (see Eq.\,\eqref{gauch}). Therefore, one should have $W(l\alpha)\Psi(g)W(l\alpha)^*=e^{-ie\sigma}\Psi(g)$, what agrees with the relations \eqref{alg3}, \eqref{alg4}. This also shows that the only dependence on gauge is through the factor appearing in the definition \eqref{alg4} of $S_\Phi$. Therefore, we interpret the additive gauge constant in $e\Phi(l)$ as a~phase variable and identify $e\Phi\sim e\Phi'\mod 2\pi$.

\section{Representations}\label{repr}

Physically meaningful representations of the algebra must satisfy some selection criteria. We impose two standard conditions: regularity of representation of $W(V)$ (which means that these unitary operators are in fact exponentiations of selfadjoint generators) and positivity of energy. The latter condition means that translations are represented by unitary transformations and their generators -- energy momentum operators -- have joint spectrum in the forward lightcone. It was shown in \cite{he98} that such representations have the following form. The representation space is $\Hc=\Hc_F\otimes\Hc_r$, where $\Hc_F$ is the usual Fock space for Dirac fermions, on which $\Psi(g)=(g,g^q)$ act in the usual way. The space $\Hc_r$ carries some regular, positive energy representation $W_r(V)$ of commutation relations \eqref{alg1}. The representation of $W(V)$ is then given by
\begin{equation}\label{Wq}
 W(V)=e^{-i\{V,V^q(+\infty,.)\}}\otimes W_r(V)\,,
\end{equation}
where $V^q(+\infty,l)$ is the quantum version of \eqref{Vplus}:
\begin{equation}\label{Vqplus}
 V^q_a(+\infty,l)=\int n^q(v) V^e_a(v,l)\, d\mu(v)\,,\quad n^q(v)=\ :\ov{g^q(v)}\ga\cdot v g^q(v):
\end{equation}
-- normal ordering in $n^q$.

Representation $W_r(V)$ is not uniquely defined. However, it is important to realize that the algebra contains in a substantial way variables at infinity. One of the consequences is that $W_r(V)$, the electromagnetic part of the representation, cannot be the usual Fock representation or any coherent state representation.
More than that, there can be no vector state with zero energy. However, a class of representations which are close to standard structures may be formed by a direct integral over coherent state representations \cite{he98}. In the following we shall use properties of these representations.

\section{Variables at spacelike infinity and quantum transformations induced by them}\label{splike}

Let us discuss the spacelike asymptotic structure of the electromagnetic field in some detail. As was seen in \eqref{spA} and \eqref{spF}, the spacelike tail is determined by $V(-\infty,l)$. In quantum case this has to be smeared with some test function, so let us consider the quantity
\begin{equation}\label{U}
 U(V^+)=\exp\Big\{\frac{i}{4\pi}\int V^+(l)\cdot V^q(-\infty,l)\,\dl\Big\}\,,
\end{equation}
where $V^+(l)$ is a smooth test function related to a homogeneous function $\vep^+(l)$ as described in Section~\ref{lgtrans}.
In spite of the smearing in $l$, the quantity \eqref{U} is not yet completely defined, as it involves a formal value at $s=-\infty$. Now, we split
\[
 V^q(-\infty,l)=V^q(+\infty,l)+[V^q(-\infty,l)-V^q(+\infty,l)]\,.
 \]
The first part involves only the Coulomb fields of outgoing particles, and in our representation has a well defined meaning \eqref{Vqplus}, giving
\begin{multline}
 U^\mathrm{Coul}(V^+)=\exp\Big\{\frac{i}{4\pi}\int V^+(l)\cdot V^q(+\infty,l)\,\dl\Big\}\otimes\id\\
 =\exp\Big\{\frac{ie}{4\pi}\int n^q(v)\int\frac{v\cdot V^+(l)}{v\cdot l}\,\dl\, d\mu(v)\Big\}\otimes \id\\
 =\exp\Big\{\frac{ie}{4\pi}\int n^q(v)\int\frac{\vep^+(l)\,\dl}{(v\cdot l)^2}\, d\mu(v)\Big\}\otimes \id\,,
\end{multline}
the third equality by an analogue of \eqref{FiV}. The second part $V^q(-\infty,l)-V^q(+\infty,l)$ involves only the free quantum outgoing field, thus it should be possible to obtain the corresponding part of $U$ by some limiting process $\dsp U^\mathrm{free}(V^+)=\id\otimes\lim_{\beta\searrow0}W_r(V^+_\beta)$, with appropriately chosen $V^+_\beta(s,l)$. This indeed is the case and we sketch the solution. For functions of $s$ we denote the Fourier transform:
\begin{equation}
 \wt{f}(\w)=\frac{1}{2\pi}\int f(s) e^{i\w s}\,ds\,.
\end{equation}
Let $\wt{h}(\w,l)$ be a smooth function, fast vanishing for $|\w|\to\infty$, such that $\wt{h}(\la^{-1}\w,\la l)=\wt{h}(\w,l)$ ($\la>0$), and $\wt{h}(0,l)=1$. Denote\footnote{What follows is a simplified version of the construction of the function $V'_{\beta a}(s,l)$ in the proof of Theorem 6.4 in \cite{he98}. The existence of the limit defining $U(V^+)$ is a special case of the construction in part (iii) of this proof.}
\begin{equation}\label{Veb}
 \wt{V}^+_\beta(\w,l)=\frac{-|\w|^{\beta-1}\wt{h}(\w,l)}{2\int|u|^{\beta-1}|\wt{h}(u,l)|^2du}\,V^+(l)\,.
\end{equation}
Then $V^+_\beta(s,l)$ is homogeneous of degree $-1$ and vanishes for $|s|\to\infty$. In consequence, as one can show, its symplectic form with any classical function $V(s,l)$ with asymptotes $V(\pm\infty,l)$ may be written as
\begin{equation}
 \{V^+_\beta,V\}=i\int \ov{\wt{\dot{V}}^+_\beta(\w,l)}\cdot\wt{\V}(\w,l)\,\frac{d\w}{\w}\,\dl\,.
\end{equation}
Now, using the definition \eqref{Veb} one finds
\begin{equation}
 \lim_{\beta\searrow0}\{V^+_\beta,V\}=-\frac{1}{4\pi}\int V^+(l)\cdot[V(-\infty,l)-V(+\infty,l)]\,\dl\,.
\end{equation}
This is a classical calculation, but remembering that $W_r(V_\beta^+)=e^{-i\{V^+_\beta,V^q\}}$ one can ask whether a similar limit exists for this operator-valued expression. It turns out that in the representations mentioned above it does: for $\beta\searrow0$ the operators $W_r(V_\beta^+)$ converge weakly to a unitary operator, which implies that the convergence is in fact in the strong operator sense. Thus we define the unitary operators:
\begin{equation}
 U^\mathrm{free}(V^+)=\id\otimes \lim_{\beta\searrow0}W_r(V^+_\beta)\,,\quad U(V^+)=U^\mathrm{Coul}(V^+)\,U^\mathrm{free}(V^+)\,.
\end{equation}

The transformation induced by these unitaries on the basic variables of the algebra is easily calculated and yields:
\begin{align}
 U(V^+)W(V_1)U(V^+)^*&=U^\mathrm{free}(V^+)W(V_1)U^\mathrm{free}(V^+)^*\\
 &=\exp\Big\{\frac{i}{4\pi}\int V^+(l)\cdot V_1(-\infty,l)\,\dl\Big\}\,W(V_1)\,,\label{lrW}\\
 U(V^+)\Psi(g_1)U(V^+)^*&=U^\mathrm{Coul}(V^+)\Psi(g_1)U^\mathrm{Coul}(V^+)^*\\
 &=\Psi(S_{\vep^+}g_1)\,,\label{lrP}
\end{align}
where $S_{\vep^+}$ is the transformation defined in \eqref{alg4}. For the interpretation of this result we note what follows. \\
1/ The asymptotic algebra is generated by the second kind-gauge invariant quantities $W(V)$ and $\Psi(g)$, which are nontrivially transformed under the action \eqref{lrW} and \eqref{lrP}.\\
2/ The action of the free part $U^\mathrm{free}(V^+)$ transforms $W(V_1)$, while leaving $\Psi(g_1)$ unchanged. Recall the interpretation of the elements of the algebra supplied by \eqref{PW}, in particular that $W(V_1)=\exp[-i\{V_1,V^q\}]$, and also note that $\frac{1}{4\pi}\int V^+(l)\cdot V_1(-\infty,l)\,\dl=-\{V_1,V^+\}$. Therefore, the rhs of \eqref{lrW} should be interpreted as $\exp[-i\{V_1,V^q+V^+\}]$, which means that this transformation adds a classical constant function $V^+(l)$ to $V^q(s,l)$. As discussed in concluding paragraphs of Section \ref{lgtrans} this addition is naturally interpreted as the asymptote of the Coulomb field of some classical external, charge-free current added by this transformation. Such addition does not change the quantum field $\Psi(g_1)$.\\
3/ The action of the Coulomb part $U^\mathrm{Coul}(V^+)$ transforms $\Psi(g_1)$, while leaving $W(V_1)$ unchanged. To interpret this transformation we choose any free Lorenz potential asymptote $V(s,l)$ (thus $V(+\infty,l)=0$) such that $\V(s,l)$ is even in $s$, and $V(-\infty,l)=V^+(l)$. This asymptote defines an even  potential $A(x)$ according to \eqref{freedl}. The scaled function $V_\la(s,l)=V(s/\la,l)$, $\la>0$, produces the scaled potential $A_\la(x)=\la^{-1}A(x/\la)$, which for $\la\to\infty$ gives the electromagnetic field with energy content tending to zero, but with unchanged spacelike tail. The operator $W(V_\la)$ creates this electromagnetic field as discussed after \eqref{alg4}. Using the algebraic relations \eqref{alg1} and \eqref{alg2} one finds that the transformation of $\Psi(g_1)$ induced by $W(V_\la)$ is identical with that induced by $U^\mathrm{Coul}(V^+)$, while $W(V_\la)W(V_1)W(V_\la)^*=\exp[-i\{V_\la,V_1\}]W(V_1)$. Integrating by parts we obtain
\begin{equation}
 \{V_\la,V_1\}=-\frac{1}{4\pi}\int V^+(l)\cdot V_1(-\infty,l)\,\dl-\frac{1}{2\pi}\int V(s/\la,l)\cdot \V_1(s,l)\,ds\dl\,,
\end{equation}
which vanishes in the limit $\la\to\infty$, as $V(0,l)=V^+(l)/2$ due to evenness of $\V(s,l)$. Now, the operators $W(V_\la)$ do not converge in this limit, but the transformation induced by them on the basic variables agrees in the limit with that induced by $U(V^+)$, which is therefore interpreted as the infrared limit transformation.\\
4/ The above discussion shows that the variables at spacelike infinity have a kind of duality structure: the Coulomb tail variables are conjugate to the infrared free tail variables.

\section{Scattering}\label{scatt}

The constructions started at the beginning of Section \ref{invstruct} and further developed up to now refer to the causal future -- outgoing fields. Similar constructions of the asymptotic algebra and its representation may be performed for causal past -- incoming fields. Smearing functions $g$ in Dirac fields $\Psi^\inc(g)$ are then as in outgoing elements $\Psi^\out(g)$. A slight modification occurs for electromagnetic fields. Recall that the smearing test functions in $W^\out(V)$ are free-field future null asymptotes $V(s,l)$ vanishing for $s\to+\infty$. Recall, also, that the past null asymptote of this field is then given by $V'(s,l)$ related to $V(s,l)$ by~\eqref{Vprim}. Therefore, we take as test functions for `in' case these past asymptotes: $W^\inc(V')$. The mapping $\Psi^\inc(g)\to\Psi^\out(g)$, $W^\inc(V')\to W^\out(V)$ is then a canonical isomorphism of the two algebras. If a complete theory could indeed be developed along the lines sketched here, then there should exist the scattering operator $S$ in the representation space such that on the level of representations $W^\out(V)=S^*W^\inc(V')S$, $\Psi^\out(g)=S^*\Psi^\inc(g)S$.\footnote{In a slightly different language -- with test functions localized in spacetime -- this was discussed in \cite{he08}.} By an appropriate limit described in the last section one obtains the long-range variables $U^\inc(V^+)$ and $U^\out(V^+)$, which are thus also related by this adjoint transformation. However, in fact the variables at spacelike infinity are conserved, so $U^\inc=U^\out\equiv U$. Therefore, one obtains
\begin{equation}
 [U(V^+),S]=0\,.
\end{equation}
This is convergent with an observation made by the authors of \cite{kps15}.

\section{Concluding remarks}

The problem of understanding, from fundamental point of view, the long-range -- infrared structure of quantum electrodynamics is still open. There are many ideas how to approach its solution, some of them mentioned above. The persistent resistance of the field to be fully understood provokes radical attempts, including a recent proposal \cite{br14} to deny the spacelike infinity quantities an experimentally accessible status.

This article summarizes a proposal going in the opposite direction (as compared with \cite{br14}). The scheme summarized above goes outside the strict local paradigm and admits into the theory some nonlocal observables at spacelike infinity. Whether it may be developed further to become indeed the asymptotic description of interacting theory is still to be seen. The scheme seems to be a rather direct quantization of the classical causal asymptotic structure, if free charged particles are to carry their Coulomb fields, not extended by some `clouds' -- a theoretical construct, rather not having experimental motivation. As we have seen, the scheme necessarily leads outside traditional structures of popular quantum field theory, like Fock space, and needs a substantial mathematical care in defining and handling quantum variables. The scheme has been further developed in other directions, including spacetime localization properties and scattering with classical currents \cite{he08, hr11}.

\section*{Appendix}
\setcounter{section}{0}
\renewcommand{\thesection}{\Alph{section}}

\section{Invariance of the integral (1)}\label{intnull}

Let the function $f(k)$ be homogeneous of degree $-2$. We note that the integral on the rhs of \eqref{d2l} may be written as
\[
 I_t=2\int f(k)\,\delta(k^2)\,\delta(t\cdot k-1)\,d^4k\,.
\]
As the Dirac distribution $\delta(k^2)$ is also homogeneous of degree $-2$, we have
\[
 \p\cdot\big[kf(k)\,\delta(k^2)\big]=0
\]
outside any neighborhood of $k=0$. Thus using the Gauss theorem we obtain for vectors $t$, $t'$ ($\theta$ is the Heaviside step function):
\[
 0=2\int\p\cdot\Big\{kf(k)\,\delta(k^2)\,\big[\theta(t\cdot k-1)-\theta(t'\cdot k-1)\big]\Big\}\,d^4k=I_t-I_{t'}\,,
\]
which shows that $I_t$ in fact does not depend on $t$ (note that the expression in braces has a compact support outside zero, so there are no boundary terms).

In fact, as mentioned in a footnote in Section \ref{measure}, the invariance of the integral is much larger. Let $\rho(k)$ be a homogeneous function of degree $+1$, such that $\rho(k)=1$ is a Cauchy surface cutting the forward lightcone (any such surface can be defined in this way in the neighborhood of the lightcone). Replacing in the last equation $t'\cdot k$ by $\rho(k)$ one shows that $I_t$ is also equal to
\[
 I_\rho=2\int f(k)\,\delta(k^2)\,\delta(\rho(k)-1)\,d^4k\,.
\]

\section{Variables $R,s,l$}\label{Rsl}

Let $B(x)$ be any differentiable field on Minkowski space. We choose a unit, future-pointing timelike vector $t$, and for $R\geq0$, $s\in\mR$ and $l\in\lc$ write

\begin{equation}
 x=\frac{st}{t\cdot l}+Rl\,,\qquad B(x)=b(R,s,l)\,.
\end{equation}
Then $b$ has the following scaling property: $b(R/\la,\la s,\la l)=b(R,s,l)$, $\la>0$. In particular:
\begin{equation}
\text{if}\ \ b(R,s,l)=\sum_{n=K}^N\frac{b^{(n)}(s,l)}{R^n}+o(R^{-N})\,,\ \text{then}\ \ b^{(n)}(\la s,\la l)=\la^{-n}b^{(n)}(s,l)\,.
\end{equation}

Differentiation: as the $R$-dependence of $b(R,s,l)$ is determined by the $s,l$-dependence and the scaling property, it is sufficient to apply $\p_s\equiv\p/\p s$ and $L_{ab}$. We find
\begin{equation}
 \p_sb(R,s,l)=\frac{t}{t\cdot l}\cdot \p B(x)\,,\quad L_{ab}b(R,s,l)=R\big(l_a\p_bB(x)-l_b\p_aB(x)\big)\,.
\end{equation}
Contracting the second equality with $t^a$ and using the first equality one finds
\begin{equation}\label{difsl}
 \p_bB(x)=l_b\dot{b}(R,s,l)+\frac{t^a}{R\,t\cdot l}L_{ab}b(R,s,l)\,,
\end{equation}
where $\dot{b}\equiv\p_sb$.

\section{Uncharged functions $V_a(l)$}\label{unchv}

We consider here the linear space of real, smooth vector functions $V(l)$ on $\lc$ with the following properties:
\begin{equation}\label{Vunch}
 V(\lambda l)=\lambda^{-1}V(l)\,,\quad l\cdot V(l)=0\,,\quad L_{[ab}V_{c]}(l)=0\,.
\end{equation}
For each function in this class there exists a smooth function $\Phi(l)$ homogeneous of degree $0$, such that
\begin{equation}\label{PV}
 L_{ab}\Phi(l)=l_aV_b(l)-l_bV_a(l)\,.
\end{equation}
Conversely, for each smooth homogeneous function $\Phi(l)$ the relation \eqref{PV} holds with some $V$.
The correspondence $V\leftrightarrow\Phi$ is not quite unique, namely:\\
1/ given $V$, the function $\Phi$ is determined up to the addition of a constant;\\
2/ given $\Phi$, the function $V$ is determined up to the following addition, referred in the following as a gauge transformation: $V(l)\rightarrow V(l)+l\alpha(l)$. To obtain $V_b$ given $\Phi$, one may extend $\Phi(l)$ into the neighborhood of the cone respecting homogeneity, differentiate $\p\Phi(l)/\p l^b$, and then again restrict the result to the cone. Different extensions yield functions differing by a gauge transformation.

However, for a given function $V$ there does exist a specific choice of the function $\Phi$, denoted in the following $\Phi_V$, satisfying the above structure and making the mapping $V\mapsto \Phi_V$ unique, namely:
\begin{equation}\label{defFiV}
 \Phi_V(l)=\frac{1}{4\pi}\int\frac{l\cdot V(l')}{l\cdot l'}\,\dl'
\end{equation}
Gauge transformations of $V$ induce then addition of constants to $\Phi_V$ by:
\begin{equation}\label{gaugeVF}
 V(l)\to V(l)+l\alpha(l)\quad \implies\quad \Phi_V(l)\to\Phi_V(l)+\frac{1}{4\pi}\int\alpha(l')\,\dl'\,.
\end{equation}
Note that not every gauge transformation changes $\Phi_V$, depending on the value of the integral in the last formula. For later use we note the following identity: if $\Phi_V(l)$ is defined by \eqref{defFiV}, then for each $t$:
\begin{equation}\label{FiV}
 \int\frac{\Phi_V(l)}{(t\cdot l)^2}\,\dl=\int\frac{t\cdot V(l)}{t\cdot l}\,\dl\,.
\end{equation}

The space of functions $V(l)$ may be equipped with the scalar product
\begin{equation}
 (V_1,V_2)_{IR}=-\int (V_1\cdot V_2)(l)\,d^2l\,.
\end{equation}
This product is strictly positive definite on the space of equivalence classes of functions $V$ up to gauge transformations. This space may be completed to a real Hilbert space, which we denote $\Hc_{IR}$ (see \cite{he98}). In the rest of this section we prove the following new result needed in the main text.
\begin{pr}
Functions of the form $\dsp V_\rho(l)=\int \frac{v\,\rho(v)}{v\cdot l}\,d\mu(v)$, where $\rho$ runs over the space of smooth, compactly supported functions on the unit future hyperboloid, such that $\int\rho(v)\,d\mu(v)=0$, form a dense subspace of $\Hc_{IR}$.
\end{pr}
\begin{proof}
Functions $V_\rho$ are smooth, and all the assumptions \eqref{Vunch} are straightforwardly checked. Let $V\in\Hc_{IR}$ be orthogonal to all functions in this class. We have to show that then $V=0$. The assumption may be written as
\begin{equation}
\int \bigg\{\int\frac{v\cdot V(l)}{v\cdot l}\,d^2l\bigg\}\rho(v)\,d\mu(v)=0
\end{equation}
for each $\rho$. The integral in braces is a continuous function in $v$, thus it follows from this equation that it is a constant function. Choosing a fixed vector $t$ one can write this as
\begin{equation}\label{Vort}
 \int V(l)\cdot \Big(\frac{v}{v\cdot l}-\frac{t}{t\cdot l}\Big)\,d^2l=0\,.
\end{equation}
For all $V$ in the Hilbert space $\Hc_{IR}$ there exist $\Phi$ satisfying relation \eqref{PV} (see \cite{he98}, Appendix). Moreover, for smooth $V_1$ one has then $(V,V_1)_{IR}=\int \Phi(l)\,\p\cdot V_1(l)\,d^2l$. Hence, Eq.\ \eqref{Vort} has an equivalent form
\begin{equation}
 \int\Phi(l)\Big(\frac{1}{(v\cdot l)^2}-\frac{1}{(t\cdot l)^2}\Big)\,d^2l=0\,.
\end{equation}
Using the freedom of adding a constant to $\Phi$ one can assume without restricting generality that $\int \Phi(l)/(t\cdot l)^2\,d^2l=0$. For all $x=\sqrt{x^2}v$ in the forward lightcone we now have
\begin{equation}
 \frac{1}{2\pi}\int\frac{\Phi(l)}{(x\cdot l)^2}\,d^2l=0\,.
\end{equation}
Now, the lhs is the integral representation \eqref{freedl} of a solution $C(x)$ of the wave equation inside the future lightcone; for $y$ inside the cone this solution has the asymptote $\lim_{R\to\infty}RC(y+Rl)=\Phi(l)/(y\cdot l)$. Therefore, $\Phi=0$, and consequently also $V=0$.
\end{proof}

\section{An identity}\label{identity}

Here we prove an identity showing the equivalence, at the classical level, of a~generator postulated by the authors of \cite{kps15} to the expression in the exponent of our definition~\eqref{U}.

Let the future null asymptote of the total electromagnetic field $F_{ab}(x)$ be given by \eqref{nullF}. Then one shows that\footnote{This follows directly from the spinor formula (2.44) in \cite{he95}.}
\begin{gather}
 \lim_{R\to\infty}R\,(x+Rl)^aF_{ab}(x+Rl)= Z_b(x\cdot l,l)\,,\\
 Z_b(s,l)=L_{ba}V^a(s,l)-V_b(s,l)+s\V_b(s,l)\,.
\end{gather}
We contract the limit relation with $t^b$ and set $x=st/t\cdot l$. This gives
\begin{equation}
 t\cdot Z(s,l)=\lim_{R\to\infty}R^2l^at^bF_{ab}\Big(\frac{st}{t\cdot l}+Rl\Big)\,.
\end{equation}
We choose a homogeneous function $\vep^+(l)$ and form the integral
\begin{equation}\label{Qept}
 Q_{\vep^+}=\frac{1}{4\pi}\int \frac{\vep^+(l)\,t\cdot Z(-\infty,l)}{t\cdot l}\,\dl\,,
\end{equation}
which is the definition put forward by the authors of \cite{kps15} as the generator of LGT. The following identity is proved by a straightforward calculation
\begin{multline}
 L_{ab}\bigg[\vep^+(l)\,\frac{t^aV^b(-\infty,l)}{t\cdot l}\bigg]=Q\bigg[\frac{\vep^+(l)}{(t\cdot l)^2}-\frac{t\cdot V^+(l)}{t\cdot l}\bigg]\\
 +V^+(l)\cdot V(-\infty,l)+\vep^+(l)\,\frac{t^a}{t\cdot l}\Big[L_{ab}V^b(-\infty,l)-V_a(-\infty,l)\Big]\,,
\end{multline}
where $Q=l\cdot V(-\infty,l)$ and $V^+$ is related to $\vep^+(l)$ as described in Section~\ref{lgtrans}; note that the last term in this identity is equal to the integrand of \eqref{Qept}. Integrating this identity with $\dl$ one finds
\begin{equation}\label{Qclass}
 Q_{\vep^+}=-\frac{1}{4\pi}\int  V^+(l)\cdot V(-\infty,l)\,\dl
\end{equation}
(the lhs of the identity is annihilated under integration by \eqref{Ld2l}, and the term proportional to $Q$ falls out by the analogue of \eqref{FiV}). The latter form of $Q_{\vep^+}$ has the advantage of being manifestly Lorentz-invariant, and it appears in the exponent of our $U(V^+)$, Eq.\,\eqref{U}.

\end{document}